\documentclass[a4paper,11pt]{article}
\usepackage[T1]{fontenc}

\usepackage{fullpage}
\usepackage{authblk}
\usepackage{dsfont}
\usepackage[makeroom]{cancel}

\usepackage[numbers]{natbib}

\usepackage{xcolor}
\usepackage{comment}

\usepackage{graphicx} 
\usepackage{subcaption}
\usepackage{researchpack} 

\usepackage{hyperref}
\usepackage[capitalise]{cleveref}

\usepackage[most]{tcolorbox}

\title{With a Little Help From My Friends: Collective Manipulation in Risk-Controlling Recommender Systems}

\author[1]{Giovanni De Toni} 
\author[2]{Cristian Consonni} 
\author[2]{Erasmo Purificato}
\author[2]{\\Emilia Gomez} 
\author[1]{Bruno Lepri} 

\affil[1]{Fondazione Bruno Kessler, Italy \texttt{\{gdetoni,lepri\}@fbk.eu}}
\affil[2]{European Commission, Joint Research Centre (JRC)\thanks{\textbf{Disclaimer}: The view expressed in this paper is purely that of the authors and may not, under any circumstances, be regarded as an official position of the European Commission.} \texttt{\{cristian.consonni,erasmo.purificato,emilia.gomez-gutierrez\}@ec.europa.eu}}

\date{\vspace{-10mm}}

\newcommand{\gdt}[1]{\textcolor{black}{#1}}

\begin{document}

\maketitle

\begin{abstract}
Recommendation systems have become central gatekeepers of online information, shaping user behaviour across a wide range of activities.
In response, users increasingly organize and coordinate to steer algorithmic outcomes toward diverse goals, such as promoting relevant content or limiting harmful material, relying on platform affordances -- such as \textit{likes}, \textit{reviews}, or \textit{ratings}.
While these mechanisms can serve beneficial purposes, they can also be leveraged for adversarial manipulation, particularly in systems where such feedback directly informs safety guarantees.
In this paper, we study this vulnerability in recently proposed risk-controlling recommender systems, which use binary user feedback (e.g., \textit{``Not Interested''}) to provably limit exposure to unwanted content via conformal risk control.
We empirically demonstrate that their reliance on aggregate feedback signals makes them inherently susceptible to coordinated adversarial user behaviour. 
Using data from a large-scale online video-sharing platform, we show that a small coordinated group (comprising only $1\%$ of the user population) can induce up to a $20\%$ degradation in nDCG for non-adversarial users by exploiting the affordances provided by risk-controlling recommender systems.
We evaluate simple, realistic attack strategies that require little to no knowledge of the underlying recommendation algorithm and find that, while coordinated users can significantly harm overall recommendation quality, they cannot selectively suppress specific content groups through reporting alone.
Finally, we propose a mitigation strategy that shifts guarantees from the group level to the user level, showing empirically how it can reduce the impact of adversarial coordinated behaviour while ensuring personalized safety for individuals.
\end{abstract}

\section{Introduction}\label{sec:intro}
Recommendation systems are used daily by nearly 50\% of the global population~\cite{bojic2024ai} for a variety of activities, ranging from entertainment to professional work.
They are also known to amplify potentially \textit{harmful content}, including \textit{fake news}~\cite{tomasel2022fakenews} and \textit{hate speech}~\cite{matamoros2021racism}, contributing to more toxic online environments~\cite{mathew2020hate} with real-world consequences such as adverse mental health outcomes \cite{golbeck2025RecsysEDrelapse, khasawneh2020examining}.
\gdt{Because these systems strongly shape what information users encounter, people often attempt to influence their behaviour through the signals they provide to the platform  - such as \textit{ratings}, \textit{reviews}, or \textit{likes}.
Users have been observed to adopt several strategies to influence recommender behaviour~\cite{jhaver2023personalizing}, often guided by informal \textit{folk theories} about how algorithms operate~\cite{morik2020controlling, eslami2016first, eslami2015always}.}
\gdt{Prior research on \textit{algorithmic collective action}~\cite{hardt2023algorithmic,sigg2025decline}} shows that these practices can scale beyond individual users, where they collectively adapt their interactions to shape future recommendations, with goals that extend beyond their personal feeds~\cite{cen2024measuring,karizat2021algorithmicfolktheories}.
\gdt{Such mobilization may arise \textit{organically}~\cite{kuchera2017anatomy} or be organized through \textit{paid labor}~\cite{lee2013crowdturfers,fayazi2015uncovering}, generating bursts of coordinated signals intended to influence content's ranking and visibility.}
Examples include attempts to boost the exposure of certain items~\cite{Marasciulo2022Anitta} or to curb the spread of harmful content by sharing screenshots rather than engaging directly with original posts~\cite{burrell2019users}.
More broadly, research has documented many grassroots and collective forms of algorithmic action in social media~\cite{burrell2019users,o2019weapons} as well as in gig-economy platforms~\cite{sigg2025decline,moralesmunoz2022spatiality} (\eg \texttt{\#DeclineNow} movement~\cite{ongweso2021doordash}) where users coordinate to induce desired positive changes in algorithmic behaviour.
\gdt{In other cases, collective interactions with platform signals take explicitly adversarial forms.
Examples include \textit{review bombing} campaigns on gaming platforms (\eg Steam), where coordinated negative feedback is used to downrank particular items following public controversies or in protest against platform policies~\cite{tomaselli2022review,kuchera2017anatomy,jiang2023ai}.
More broadly, coordinated manipulation of platform signals has also been used to strategically influence information ecosystems, for instance, through \textit{online astroturfing} campaigns~\cite{schoch2022coordination} that attempt to amplify political narratives or misinformation artificially (\eg 2016 U.S. presidential elections~\cite{bail2020IRAelections}).}

\gdt{
Taken together, these examples illustrate how collective behavior can meaningfully reshape the signals that recommendation algorithms rely on.}
\gdt{More importantly, users' signals are not only used to rank content but are increasingly central to efforts to mitigate the \textit{harms} associated with recommendation systems (\eg X's Community Notes).}
\gdt{For example}, a common approach has been to provide users with interface-level controls, such as a \textit{``Not Interested''} button, that enable them to signal dissatisfaction with an item directly to the platform~\cite{hong2025social}.
\gdt{However, the same signals may also be strategically manipulated by coordinated groups of users.
As a result, affordances intended to improve user safety and content quality may inadvertently create \textit{new levers} for influencing recommendation outcomes.}
In this paper, we investigate the following question:
\begin{center}
\textit{What would happen if a group of users organized themselves to strategically use the ``Not Interested'' feature to alter recommender behavior?}
\end{center}
\gdt{Specifically, we investigate how coordinated user behavior may interact with the guarantees provided by \textit{risk-controlling recommender systems}~\cite{detoni2025you}.}
They are a recently proposed recommender system that directly leverages users’ negative feedback (the “\textit{Not Interested}” signal) to bound, in expectation, the frequency of unwanted items in users’ feeds via conformal risk control~\cite{angelopoulos2024conformal}.
They introduce a filtering layer that excludes potentially harmful or unwanted items based on a \textit{global threshold}, calibrated on a held-out set of historical user-item interactions.
In contrast to standard interface-level controls, these \textit{risk-controlling recommender systems} provide the first formal guarantees that directly link user actions, such as reporting an item, to deterministic changes in recommendation outcomes in expectation.

\gdt{Perhaps unintuitively, in our paper}, we show that adversarial coordinated behaviour can only \textit{strengthen} the formal risk-control guarantees: adversarial users can artificially increase the empirical risk observed during calibration, which in turn forces the system to adopt \textit{more conservative} filtering thresholds.
However, we demonstrate that this seemingly beneficial effect may come at a \textit{significant cost to recommendation quality}.
In particular, coordinated reporting can induce \emph{disproportionate degradations} in system performance, even when carried out by very small collectives.
For example, a group of just 40 adversarial users, reporting at most 1\% of the items they encounter, may reduce standard nDCG by up to 20\% \gdt{in the worst case scenario}.
Interestingly, although changes to the risk-control guarantees scale \textit{linearly} with the number of coordinated users, our results show that even small collectives can make it substantially more difficult to lower the expected risk experienced by non-adversarial users at test time.
This asymmetry is consistent with prior work \gdt{postulating} that, in settings dominated by weak signals, such as binary ``\textit{Not Interested}'' feedback\footnote{For example, in a dataset from the large-scale video-sharing platform \textit{Kuaishou} \cite{gao2022kuairand}, user-provided “\textit{Not Interested}” feedback is extremely sparse, with an average reporting rate of approximately $0.002\%$ per user~\cite{detoni2025you}.}, coordinated users who inject \textit{stronger signals} can exert outsized influence \cite{hardt2023algorithmic}.
We also show that adversarial users cannot selectively suppress the exposure of a target item group (\ie the frequency with which items from that group appear in users’ top-$k$ recommendations), since the underlying risk-controlling procedure is agnostic to group membership\footnote{Nevertheless, risk-controlling recommender systems may still induce \emph{disparate impacts} when the underlying recommender exhibits pre-existing biases toward particular item groups (\cf \cref{app:disparate-impact}).}.
\gdt{Lastly, we describe and empirically validate a simple improvement that makes risk-controlling recommender systems robust to adversarial manipulation by enforcing guarantees at \textit{individual} rather than population level.}

\gdt{Our work can be viewed as a \textit{pre-deployment audit} of risk-controlling recommender systems, examining how coordinated user behaviour may interact with their guarantees.
Indeed, several regulations and standardized frameworks propose \textit{auditing} as a mechanism to identify and mitigate potential risks associated with artificial intelligence systems~\cite{lam2024framework,ftc2022commercialsurveillance,eu2022dsa,union2021proposal}.
For instance, the EU’s \textit{Digital Services Act} (DSA)~\cite{eu2022dsa} requires very large online platforms to conduct systemic risk assessments that must account for risks arising from \textit{intentional manipulation} of the service, including \textit{coordinated} or automated activity (Article 34(2) of the EU DSA).
Motivated by this perspective, we argue that evaluating such dynamics should be part of the preliminary assessment of risk-controlling mechanisms before their widespread deployment in real-world platforms.}

\paragraph{Our contributions.}
Formally, we first theoretically characterise how the coordinated users who \textit{strategically} employ the \textit{``Not Interested''} feedback mechanism affect the risk-control guarantees of recommender systems (\cref{theorem:adv-impact}).
Building on these results, we analyse a set of realistic \textit{reporting strategies}\footnote{Throughout the paper, we use the term ``report'' to indicate a feedback given to an item through a ``\emph{Not Interested}'' or ``\emph{Show me less of this}'' button. This should not be confused with reporting an item in terms of flagging content for violating the terms of service, community standards, or other norms of the platform, an action that triggers immediate removal or review of the item that is usually reserved for specific categories of content.} that collectives may adopt to maximise their influence on recommendation outcomes (\cref{sec:reporting_strategies}).
Second, we empirically evaluate the impact of such collectives on a risk-controlling recommender system~\cite{detoni2025you} using real-world interaction data from a large online video-sharing platform, \textit{Kuaishou} (\cref{sec:evaluating_effects_of_collective}).
Our evaluation focuses on standard recommendation metrics, including \textit{nDCG} and \textit{Recall}, as well as on changes in \textit{item exposure}.
\gdt{Finally, we empirically show how calibrating a risk-controlling recommender system with an individual-level threshold can successfully mitigate adversarial collective behaviour (\cref{sec:mitigating-coll-action})}.

\section{Related Work}
\label{sec:related_work}

\gdt{Our paper builds upon further related work on algorithmic collective action, adversarial attacks on recommender systems, and countering the harmfulness of recommender systems.}
Despite the varying interpretations of coordinated user behavior, ranging from \textit{beneficial} \cite{hardt2023algorithmic} to \textit{adversarial} \cite{lee2013crowdturfers}, we use the terms \textit{``collective''}, \textit{``coordinated groups''}, \textit{``strategic users''}, and \textit{``adversaries''} interchangeably throughout this paper.

\paragraph{Collective Action on Algorithmic Systems.}
\gdt{
\citet{hardt2023algorithmic} initiated the theoretical study of \textit{algorithmic collective action} by employing a simple model to describe how the coordinated efforts by groups of users can steer model outcomes.
\citet{sigg2025decline} further studied a combinatorial model to describe the strategic interaction between workers and a gig economy platform.}
\gdt{Several other works expanded the study of algorithmic collective action in general recommendation systems~\cite{mendler2024engine,baumann2024algorithmic,karan2025algorithmic,karan2025sync}.
For example, by proposing frameworks to study the interactions between different groups with distinct objectives \cite{karan2025algorithmic,karan2025sync} underlying how the collective size, rather than the homogeneity or heterogeneity of the group, is the more critical factor in how effectively a collective can manipulate recommender system outcomes.
Our work is the first to study the effects of collectives on affordances designed to mitigate the spread of harmful content on recommender systems.}
Specifically linked to our work is the study by \citet{baumann2024algorithmic}, which investigated collective action in the context of music recommendation.
They demonstrated that a small collective of fans (controlling as little as 0.01\% of training data) could significantly amplify the visibility of an underrepresented artist by strategically reordering their playlists.
\gdt{Our setting is conceptually distinct. We study the effects of users interacting with a \textit{``Not Interested''} affordance and evaluate broader platform-level consequences, including recommendation performance and item exposure.
In addition, we focus on a \textit{risk-controlling recommender system}~\cite{detoni2025you}, whereas they consider a classical recommender model.
Finally, unlike \citet{baumann2024algorithmic}, we do not retrain the recommender.
In our experiments, the recommender remains fixed, and only the risk calibration threshold is recomputed, an intervention that is substantially less costly for the platform (\cf \cref{sec:preliminaries}).}

\paragraph{Collective Adversarial Activities in Recommender Systems.}

\gdt{The literature on \textit{``astroturfing''} \cite{schoch2022coordination} or \textit{``crowdturfing''} activities \cite{tricomi2023we,lee2013crowdturfers,wang2012serf} shows how it is possible to cheaply acquire fake engagement on certain gig economy platforms (\eg Fiverr, Sina Weibo) - often mediated by low-wage workers hired to create fake grassroots support or spam, such as \textit{fake reviews}, \textit{comments}, or \textit{likes}. 
For example, studies show that it costs around 0.80-3.00\$ to hire 100 real users who will interact with the target profiles to boost their engagement rate~\cite{tricomi2023we}.
These profiles, being \textit{real users}, are harder to detect by online platforms, and they represent a potential source of malicious collective action~\cite{wang2012serf}.
Our work considers such a real-world setting, where a malicious group of users organizes to collectively exploit the risk-controlling guarantees of a platform towards altering its behaviour.}
Furthermore, our work is technically close to the concept of \textit{shilling attacks}~\cite{zhang2025robust} -- also known as  
\textit{profile-injection attacks} -- where malicious actors create and inject fake user profiles, user interactions, or ratings.
For example, \citet{anelli2021adversarial} illustrated how an adversary might modify audio signals in a music recommender system to associate tracks with explicit and illegal content, linking adversarial perturbations to harmful recommendations.
\gdt{For these reasons, we also provide an empirically validated strategy that can mitigate the effect of adversarial collective actions on simulation studies.}
Our discussion here on adversarial attacks is necessarily limited, so we refer to further surveys~\cite{fan2022comprehensive,wang2024trustworthy} highlighting how adversarial attacks are capable of influencing public opinion by spreading fake news and manipulated content through recommender systems.

\paragraph{Countering the Harmfulness of Recommender Systems.}
Interest in understanding how to deal with \textit{harmful content of recommender systems} in online platforms has spanned the last decade.
\gdt{Beyond harmful material, recommenders can also disseminate content that users perceive as \textit{unwanted}, intended as what may conflict with the users' personal values, trigger negative emotions, or may be essentially mediocre, thus perceived as a \textit{``waste of time''}~\cite{hong2025social,liu2024train}.}
\citet{gillespie2022not} described how social media platforms reduce the visibility and reach of potentially dangerous content by demoting it in algorithmic rankings and recommendations, rather than removing it entirely.
By including a negative signal to prevent specific content from being recommended to users, this study shows how platforms can use recommenders as an effective means of targeted content moderation.
Along this line, various concrete strategies have been proposed to tackle harmful content in recommender systems~\cite{chee2024harm,liu2024train,arora2023detectingharm}.
\gdt{Unfortunately, evidence suggests that these mechanisms are often ineffective since they usually suffer from low visibility~\cite{ibrahim2024characterizing} and are opaque in how they affect the recommendation process, particularly beyond the user’s own feed~\cite{wang2025end,hong2025social}.
For instance, users often expect negative feedback to reduce the prevalence of similar content, without penalising the content itself or its creator~\cite{hong2025social}.}
Recently, \citet{detoni2025you} introduced \textit{risk-controlling recommender systems}, the first model-agnostic method using conformal risk control to provably bound unwanted content in personalised recommendations with distribution-free and model-free guarantees.
\gdt{In the following sections, we will investigate and audit such \textit{risk-controlling recommender systems} by testing their effect in the presence of a collective group of users acting adversarially.}
\section{Preliminaries}
\label{sec:preliminaries}

Let us now formally outline how to build a \textit{risk-controlling recommendation system}.
Consider a finite set of items $\calI = \{i_1, \ldots, i_N\}$ and users $\calU = \{u_1, \ldots, u_M\}$.
Without loss of generality, we assume that each item (or user) is described by a $d$-dimensional feature vector.
Let us consider a ranker $f: \calU \times \calI \rightarrow \mathbb{R^+}$ that estimates how \textit{relevant} an item is, based on the user's profile and the item's characteristics.
Given the ranker and a target user $u \in \calU$, we usually want to pick the best $k \in \bbN^+$ items to display to the user the recommendation set $\calS(U=u, k) = \{ i_j \in \calI : \pi(i, u) \leq k \}$, where $\pi: \calI \times \calU \rightarrow \mathbb{N}$ returns an item ordering induced by the ranker scores, \eg from the highest to the least preferred item.  
In standard \textit{learning-to-rank} tasks, we want to find the optimal ranker $f^* = \argmax_{f \in \calF} \ell(S(U,k))$ that maximises a target metric $\ell: \calI \times \calU \rightarrow \mathbb{R}$, such as the \textit{engagement}.

Let $R: 2^{|\calI|} \times U \rightarrow [0,1]$ be a function that describes the \textit{risk} of giving a recommendation set $\calS \subseteq \calI$ to a user.
For example, the function can model the \textit{unwantedness} of the recommendations, such as the fraction of items flagged as ``\textit{Not Interested}'' by a user. 
Further, consider $r: \calI \times \calU \rightarrow [0,1]$ to be a score function that quantifies the risk of a single item (\eg the probability that a user will flag an item), learned from historical user-item interactions.
Adopting a common two-stage setup \gdt{used in many commercial recommender systems~\cite{huang2025comprehensive}}, we filter items based on their score, \gdt{removing them only if it is} greater than a threshold $\lambda \in \Lambda$, before building the final recommendation set:
\begin{equation}
    \calS_\lambda(U=u, k) = \{ i_j \in \calT_{\lambda}(U=u) : \pi(i, u) \leq k \} \quad \text{where} \quad \calT_{\lambda}(U=u) = \{ i \in \calI : 1-r(i, u) \geq \lambda \}
    \label{eqn:recommendation-set-filter}
\end{equation}
\gdt{where $\calT_{\lambda}(U=u)$ denotes the set of candidate items whose score is above the threshold.}
\gdt{In the classical conformal risk control formulation~\cite{angelopoulos2024conformal}, the filtering condition is typically written as $r(i,u) \geq \lambda$. 
Our formulation, $1 - r(i,u) \geq \lambda$, is equivalent under the assumption that $r(i,u) \in [0,1]$, and selects items with lower predicted risk (\ie lower probability of being flagged as harmful).}
In summary, we want to find the optimal threshold $\lambda$ that maximises the given objective, such as the \textit{engagement}, by keeping the target \textit{risk} below a platform-defined level $\alpha \in [0,1]$ in expectation:
\begin{equation}
    \lambda^* = \argmax_{\lambda \in \Lambda} \bbE[\ell(S_\lambda(U, k))] \qquad \text{s.t.} \qquad \bbE[R(S_\lambda(U, k))] \leq \alpha
    \label{eqn:risk-controlling-set}
\end{equation}
\citet{detoni2025you} showed that we can provably control the risk in recommender systems, as defined by Equation~\ref{eqn:risk-controlling-set}, using simple \textit{binary} user feedback, captured through ``\textit{Not Interested}'' interactions by exploiting \textit{conformal risk control}~ \cite{angelopoulos2024conformal}.
Under standard assumptions, including data exchangeability and mild regularity conditions on the risk function (\eg monotonicity), they show that a valid filtering threshold can be selected as $\hat{\lambda} = \inf \left\{ \lambda : \frac{Q}{Q+1}\hat{R}(\lambda) + \frac{1}{Q+1} \leq \alpha \right\}$, where $\hat{R}(\lambda) = \frac{1}{Q}\sum_{j=1}^{Q} R(S_\lambda(U = u_j, k))$ is the empirical risk estimated on a \textit{calibration set}, a held-out collection of $Q$ user-item interactions used exclusively to guarantee risk control.
They formalise the risk of a recommendation set $S_\lambda(U, k)$ as:
\begin{equation}
R(S_\lambda(U, k)) = \frac{|{ i \in S_\lambda(U, k) : H(U, I = i) = 1 }|}{|S_\lambda(U, k)|}
\label{eqn:set-based-harmfulness}
\end{equation}
where $H$ is a binary random variable, drawn from $P(H \mid I, U)$, indicating whether an item is deemed unwanted by a user ($H = 1$).
To ensure that the recommendation set remains of size $k$ (or close to it), despite potentially aggressive filtering, they propose filling filtered positions with previously consumed items that have not been flagged as unwanted, referred to as \textit{repeated safe items}. 
More broadly, conformal risk control is both \textit{model-free} and \textit{post hoc}: it can be applied on top of any pretrained ranker $f$ or risk predictor $r$ that follows the two-stage architecture, requiring minimal additional assumptions and integration effort.

\paragraph{Modeling Coordinated Attacks in Risk-Controlling Recommender.}
\gdt{We consider an adversarial variant of algorithmic collective action \cite{hardt2023algorithmic,baumann2024algorithmic}, in which a subset of users coordinates to strategically manipulate system inputs with the \textit{explicit goal of degrading or suppressing the exposure of targeted content} (\eg silencing a creator or reducing the visibility of a content group).
Unlike prior work, where collectives pursue socially motivated or welfare-improving objectives, we focus on adversarial objectives that are strategically aligned but not necessarily \textit{ethically grounded}.
Such collectives can plausibly emerge in online environments where users share common interests or antagonisms, including loosely organized communities (\eg Reddit, 4chan, 8chan) or more structured campaigns.
Empirical evidence shows that large-scale coordinated actions -- ranging from attention manipulation (\eg \texttt{\#TulsaFlop}~\cite{bandy2020tulsaflop}) to harassment campaigns (\eg \texttt{\#GamerGate}~\cite{mortensen2020negative}) -- can arise with minimal central coordination, often relying on broadcast communication and shared intent rather than tight synchronization.
Moreover, low-cost crowdturfing infrastructures \cite{lee2013crowdturfers} further reduce coordination barriers by enabling centrally orchestrated manipulation through paid participants.
In our setting, these mechanisms translate into users collectively marking targeted items as \textit{``Not Interested''}, thereby providing systematically biased feedback to the recommender system.
We assume a weak coordination model: participants share a common target and action (\eg flagging a predefined set of items), but do not require fine-grained information about other users, or real-time synchronization.
This captures both \textit{decentralized campaigns} and \textit{centrally coordinated efforts}.
Formally, given a calibration set of $Q$ users, we assume that a subset $K < Q$ behaves adversarially\footnote{In practice, a risk-controlling recommender system must calibrate its filtering threshold $\lambda$ using the full user population to avoid distribution shifts~\cite{angelopoulos2024conformal} that could happen at test time.
\gdt{Moreover, it has been shown that some platforms have more relaxed policies, thus allowing users to create potentially dozens of new accounts~\cite{zeng2022content}.}
As a consequence, adversarial users are likely to be included in the calibration set by default.}, yielding a fraction $\beta = \frac{K}{Q}$.
Finally, we formalize the collective’s objective along two distinct dimensions: (i) the degradation of overall recommendation performance, measured by the decline in recommendation quality experienced by non-adversarial users at test time, and (ii) the targeted suppression of specific groups of items, reflected in reduced exposure.}

\section{Manipulating a Risk-Controlling Recommendation System}

To identify the effective targets of adversarial behaviour, within the theoretical framework of conformal risk control, we first analyse the impact of a coordinated group of adversarial users on the resulting upper bound on risk.
Unlike prior work on adversarial conformal prediction and risk control \cite{zargarbashi2024robust,angelopoulos2024conformal,gendler2021adversarially} or classical \textit{shilling attacks} \cite{zhang2025robust}, the adversaries we consider do not inject forged or poisoned data (\eg by manipulating items or interactions).
Conversely, they operate entirely within the platform’s intended affordances, strategically using standard feedback mechanisms.
We begin by presenting our main theoretical result, followed by a discussion of its key practical implications.
The complete proof is available in \cref{app:proofs}. 

\begin{theorem}

Consider a held-out calibration set $\calQ = \{(u,i,h)_j\}_{j=1}^Q$.
Let us assume that there are $K$ users within $Q$ that are behaving strategically.
Given a target level $\alpha \in [0,1]$, let us consider $\hat{\lambda} \in \Lambda$ choosen as $\hat{\lambda} = \inf \left\{ \lambda : \frac{Q}{Q+1}\hat{R}(\lambda) + \frac{1}{Q+1} \leq \alpha \right\}$ where $\hat{R}(\lambda) = \frac{1}{Q}\sum_{j=1}^Q R(S_\lambda(U=u_j, k))$ is computed over the calibration set. 
Lastly, let us denote the non-adversarial users with $U_{nonadv}$.
Then, under data exchangeability, we have that the expected risk for non-adversarial users is bounded from above by:
\begin{equation}
    \bbE[R(S_\lambda(U_{nonadv}, k))] \leq \max\{0, \alpha -\frac{K}{Q+1}r_\lambda^{adv}\}
\end{equation}
where \gdt{$r_\lambda^{adv} = \frac{1}{K}\sum_{u_{adv} \in \calK} R(S_\lambda(U=u_{adv}, k))$ is the risk of the adversarial users $u_{adv} \in \calK \subset \calQ$ within the calibration set.}
\label{theorem:adv-impact}
\end{theorem}
Intuitively, \cref{theorem:adv-impact} shows that each adversarial user consumes a portion of the available risk budget (\cf $\tfrac{K}{Q+1} r^{\text{adv}}_\lambda$), thereby pushing the actual risk reduction for unsuspecting users closer to zero. 
This effect scales \textit{linearly} with the number of adversarial users included in the calibration set.
As a consequence, when non-adversarial users request a moderate reduction in unwanted content (\eg a 20\% decrease), the system may be forced to enforce substantially \textit{more aggressive} filtering, potentially eliminating nearly all risky items, to satisfy the global risk constraint in the presence of adversarial feedback.
Finally, under the extreme assumption that adversarial users flag every item in their recommendation sets as unwanted, we obtain a tighter upper bound:
\begin{corollary}
    Let us assume that $R(S_\lambda(U_{adv}, k)) = 1$ for each adversarial user with the calibration set. 
    Then, the expected risk for non-adversarial users $U_{nonadv}$ is bounded from above:
    \begin{equation}
        \bbE[R(S_\lambda(U_{nonadv}, k))] \leq \max\{0, \alpha -\frac{K}{Q+1}\}
    \end{equation}
    \label{corollary:strong-adv}
\end{corollary}
In summary, as shown by \cref{theorem:adv-impact} and \cref{corollary:strong-adv}, the actual expected risk experienced by standard users decreases linearly based on the number of adversarial users.
While this may appear beneficial at first glance, in \cref{sec:evaluating_effects_of_collective} we demonstrate that such reductions can come at a substantial cost to recommendation quality.
Importantly, \cref{theorem:adv-impact} highlights a key consideration for coordinated collectives.
For a fixed collective size, the magnitude of the effect depends critically on the empirical adversarial risk $r^{\mathrm{adv}}_\lambda$ observed in the calibration set.
As the filtering threshold increases, the influence of adversarial users may diminish, for instance, once all items flagged by the collective are removed from the top-$k$ recommendations.
Consequently, an effective adversarial strategy must sustain a high empirical risk at calibration time, as the threshold increases, thereby preserving its impact on the system.

\section{On the Reporting Strategies}
\label{sec:reporting_strategies}

Building on the previous sections, we observe that the effectiveness of an adversarial collective hinges on its ability to adopt a \textit{reporting strategy} that maintains a high empirical risk during calibration.
We begin by noting that the filtering procedure defined in \cref{eqn:recommendation-set-filter} removes items in order of decreasing risk score $1-r(i,u)$, replacing those deemed most risky first.
As a result, an effective strategy for adversarial users is to report items that have \textit{low estimated} risk scores $r(i,u)$.
Such items are intuitively filtered only at higher thresholds $\lambda$, thereby forcing the system to adopt more conservative filtering to satisfy the global risk constraint.
We formalise this intuition through the following reporting strategy:
\begin{equation}
    \texttt{LowRisk}(\gamma) = \{ i \in \calI_u : \pi_r(i, u) \leq \lceil \gamma \cdot |\calI_u| \rceil \}
\end{equation}
where $\mathcal{I}_u$ denotes the set of items shown to user $u$, $\pi_r(i, u)$ induces an ordering of items based on their estimated risk scores $1-r(i,u)$, and $\gamma \in [0,1]$ specifies the fraction of items that an adversarial user reports to the platform.
While effective in principle, the \texttt{LowRisk}$(\gamma)$ strategy assumes white-box access to the recommender’s internal risk scores -- an assumption that is typically unrealistic in practice -- \gdt{thus representing the worst-case theoretical scenario.}
\gdt{Nevertheless, prior work has shown the feasibility of data-free model extraction attacks against sequential recommenders~\cite{yue2021blackboxattacks}, thus making this strategy potentially plausible for organized collectives.}
More generally, adversarial strategies should rely on \textit{minimal background knowledge}, while maximising their downstream impact, and also satisfying practical constraints such as remaining difficult to detect.
Because risk scores are hidden within the platform backend, we therefore turn to alternative strategies based on \textit{observable proxies} available through the user interface (\eg item popularity, ranking position, or metadata), which adversaries could plausibly exploit to approximate the \texttt{LowRisk} strategy.

\paragraph{Further reporting strategies.}
Motivated by documented cases of real-world algorithmic collective action \cite{sigg2025decline}, we consider three additional plausible adversarial reporting strategies targeting \textit{risk-controlling recommender systems}.
Given that the risk score of each item is unavailable to the platform's users, these alternative strategies must rely solely on information that the adversaries can readily garner from the platforms, respectively: (i) the number of \textit{likes} an items has received (\texttt{Likes}); (ii) an item’s rank position (\texttt{TopRank}); and (iii) the item's assigned \textit{tags} (\texttt{Tags}).
The intuition behind each strategy is described as follows.
For the \texttt{TopRank} strategy, the idea is that an item’s rank position plays a role in the risk-control procedure: given a risk score, higher-ranked items are more likely to be replaced than lower-ranked ones. Therefore, reporting higher-ranked items may delay their removal and require a larger filtering threshold, keeping $r_\lambda^{adv}$ high.
For the \texttt{Likes} strategy, item popularity, as proxied by the number of likes, may inversely correlate (imperfectly) with perceived unwantedness; therefore, the idea is that highly liked items are, in expectation, less likely to be flagged as unwanted and therefore tend to have lower risk scores.
Finally, for the \texttt{Tags} strategy, we assume that each item belongs to a group $g \in \mathcal{G}$, where $\mathcal{G}$ denotes a platform-defined taxonomy.\footnote{For simplicity, we assume a single categorical assignment per item. Nevertheless, we acknowledge potential challenges arising when considering \textit{intersectionality} aspects in recommender systems evaluation \cite{pmlr-v81-ekstrand18b}.}
For example, on a video-sharing platform, items may be categorised as \textit{entertainment} or \textit{cooking}. Some groups may be systematically more popular than others and, by the same reasoning, associated with lower expected risk. 
Furthermore, we consider this strategy to represent the scenario where a collective wants to target \textit{all the items} of a specific category with the intention of reducing the visibility of that category on the platform for all users.
Based on these observations, we define the following reporting strategies:
\begin{gather}
    \texttt{Likes}(\gamma) = \{ i \in \calI_u : \pi_{l}(i, u) \leq \lceil \gamma \cdot |\calI_u| \rceil \} \\
    \texttt{TopRanker}(\gamma) = \{ i \in \calI_u : \pi_{f}(i, u) \leq \lceil \gamma \cdot |\calI_u| \rceil \} \\
    \texttt{Tag}(g) = \{ i \in \calI_u \cap \calI_g\}
\end{gather}
where $\pi_{l}(i, u)$ and $\pi_{f}(i, u)$ denote the item orderings induced, respectively, by the number of likes and by the recommender’s ranking function, and $\mathcal{I}_g$ denotes the set of items belonging to group $g \in \mathcal{G}$.
As a benchmark, we consider a naive strategy where the collective of adversarial users simply decides to report a fraction of items chosen at random:
\begin{equation}
    \texttt{Random}(g) = \{ i \in \calI_u \} \quad \text{s.t.} \quad |\texttt{Random}(g)| = \lceil \gamma \cdot |\calI_u| \rceil
\end{equation}

\section{Evaluating the Effects of Collective Action}
\label{sec:evaluating_effects_of_collective}

In this section, we empirically examine the impact of coordinated adversarial users who seek to alter recommender behaviour by strategically flagging videos as \textit{unwanted}.
Our analysis leverages real-world user-item interaction data and a state-of-the-art recommender model.
Further, we released all the source code, raw data, and pre/post processing script on GitHub under a permissive license.\footnote{\url{https://github.com/geektoni/collective-action-recsys}}
Building on the theoretical insights developed in the previous sections, we address the following research questions:
\begin{itemize}
\item[\textbf{RQ1}] Which reporting strategy is most effective at sustaining high adversarial risk $r_\lambda^{adv}$ across all $\alpha \in [0,1]$?
\item[\textbf{RQ2}] How does coordinated adversarial reporting affect standard performance metrics?
\item[\textbf{RQ3}] To what extent can adversarial collectives influence the exposure of specific item groups through reporting?
\end{itemize}

\paragraph{Experimental details.}
We conduct our experiments using the \textit{KuaiRand} dataset \cite{gao2022kuairand}, which consists of real user-item interactions collected from the large-scale video-sharing platform \textit{Kuaishou}.\footnote{\url{https://www.kuaishou.com/}. Last accessed: 13/01/2026.}
To the best of our knowledge, KuaiRand is the only publicly available dataset that includes realistic \textit{``Not Interested''} feedback provided by users in response to recommended videos.
Our experimental setup follows the two-stage architecture proposed by \citet{detoni2025you}, using their pretrained LightGCL\footnote{We focus on LightGCL for risk prediction, as it achieved the strongest performance in the original study. By strongest performance, we mean that LightGCL \gdt{uses} \textit{less} repeated safe items \gdt{than comparable models} to achieve the same risk level~\cite[Figure 4]{detoni2025you}.} recommender model \cite{caisimple}.
We adopt the same calibration and test splits as in the original work to ensure comparability.
From the calibration set, we randomly sample a fraction $\beta \in (0.001, 0.1)$ of users to form an adversarial collective.
Users in this group are assumed to coordinate and follow one of the reporting strategies described in \cref{sec:reporting_strategies}.
For the \texttt{LowRisk}, \texttt{Likes}, \texttt{TopRanker}, and \texttt{Random} strategies, we additionally specify a reporting rate $\gamma \in [0,1]$, representing the fraction of videos each adversarial user flags as ``\textit{Not Interested}.''
For the \texttt{Tag} strategy, we choose the top-3 most frequent tags in the \textit{KuaiRand} dataset $g \in \{39, 34, 67\}$. 
In practice, we restrict $\gamma$ to $\{0.001, 0.01, 0.1\}$, reflecting the assumption that users who report an excessively large fraction of items would likely be detected and removed by the platform~\cite{arora2023detectingharm}.

\begin{figure}
    \centering
    \begin{subfigure}[t]{0.7\textwidth}
        \centering
        \includegraphics[width=\linewidth]{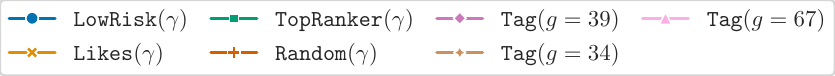}
        \vspace{0.1em}
    \end{subfigure}
    \begin{subfigure}[t]{0.32\textwidth}
         \centering
            \includegraphics[width=\linewidth]{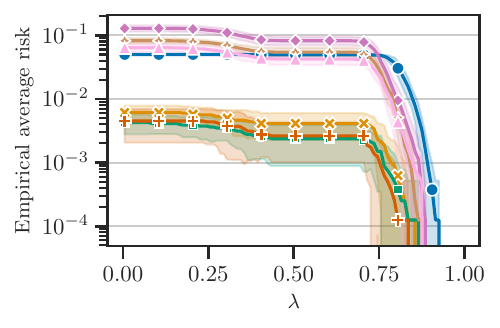}
          \caption{$\gamma=0.001$}
          \label{fig:reporting_strategies_gamma_0.001}
     \end{subfigure}
    \begin{subfigure}[t]{0.32\textwidth}
        \centering
        \includegraphics[width=\linewidth]{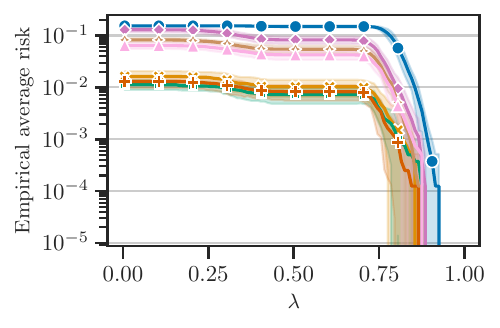}
        \caption{$\gamma=0.01$}
     \end{subfigure}
     \begin{subfigure}[t]{0.32\textwidth}
         \centering
            \includegraphics[width=\linewidth]{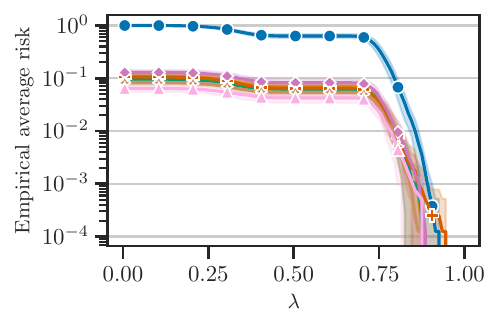}
          \caption{$\gamma=0.1$}
          \label{fig:reporting_strategies_gamma_0.1}
     \end{subfigure}
     \caption{
        Expected empirical risk of adversarial users at calibration time as a function of the filtering threshold $\lambda \in [0,1]$, for reporting rates $\gamma \in \{0.001, 0.01, 0.1\}$ and different reporting strategies.
        The collective size is fixed to $\beta = 0.01$. 
        Shaded areas indicate one standard deviation over 10 runs. 
     }
     \label{fig:reporting_strategies}
\end{figure}

\paragraph{Evaluation metrics.}
We evaluate recommendation quality using nDCG@k and Recall@k, fixing $k = 20$ throughout, as is common in the literature \cite{detoni2025you, caisimple}.
For a given risk level $\alpha \in [0,1]$, and for each reporting strategy and metric $m \in \{\text{nDCG@}k, \text{Recall@}k\}$ we quantify the impact of adversarial behaviour through the normalised performance change at test time:
\begin{equation}
\text{Reduction}(\beta) = \frac{1}{\beta}\bigl(m(0) - m(\beta)\bigr)
\label{eqn:reduction-ATE}
\end{equation}
where $m(0)$ denotes the performance of the unperturbed risk-controlling recommender system when enforcing risk level $\alpha$ (cf.\ Eq.~\ref{eqn:risk-controlling-set}), and $m(\beta)$ denotes performance based on the fraction of adversarial users within the calibration set $\beta = \frac{K}{Q}$.
\gdt{\cref{eqn:reduction-ATE} yields a dimensionless quantity, and its value is relative to the reduction theoretically expected for a collective of size $\beta$.}
A reduction of 0 indicates no observable effect, whereas a reduction of 1 corresponds to an effect proportional to the size of the collective. Values greater than 1 indicate a \textit{disproportionate} impact.
To address RQ3, we additionally measure group-level exposure, defined as the probability that an item belonging to group $g$ appears in the top-$k$ recommendations:
\begin{equation}
\mathrm{Exposure}(k, \beta) = \mathbb{E}\bigl[\mathds{1}\{ I \in S_{\lambda^{*}_{\alpha,\beta}}(U, k) \} \mid G = g \bigr]
\label{eqn:expected-group-exposure}
\end{equation}
where $\lambda^{*}_{\alpha,\beta} \in \Lambda$ denotes the optimal threshold for risk level $\alpha$ in the presence of an adversarial collective of size $\beta$.

\begin{figure}
    \centering
    \begin{subfigure}[t]{0.7\textwidth}
        \centering
        \includegraphics[width=\linewidth]{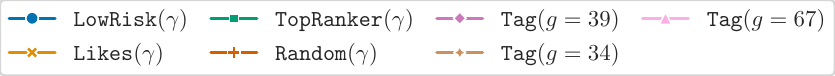}
        \vspace{0.1em}
     \end{subfigure}
     \begin{subfigure}[t]{0.43\textwidth}
        \centering
        \includegraphics[width=\linewidth]{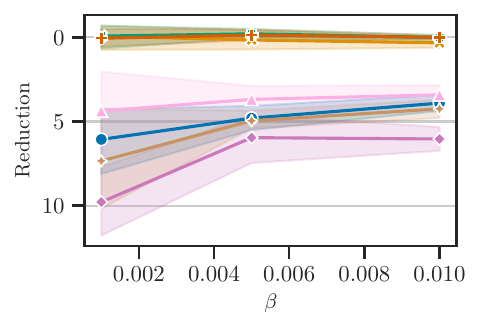}
        \caption{nDCG @ 20 ($\gamma=0.001$)}
        \label{fig:ndcg_0.001}
    \end{subfigure}
    \begin{subfigure}[t]{0.43\textwidth}
        \centering
        \includegraphics[width=\linewidth]{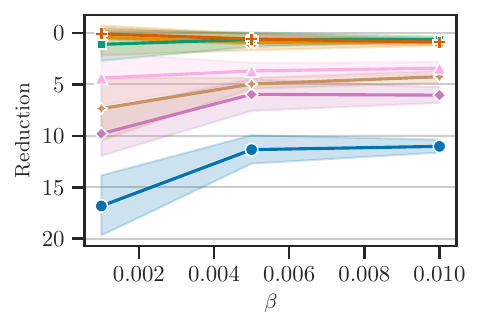}
        \caption{nDCG @ 20 ($\gamma=0.01$)}
    \end{subfigure}
    \begin{subfigure}[t]{0.43\textwidth}
        \centering
        \includegraphics[width=\linewidth]{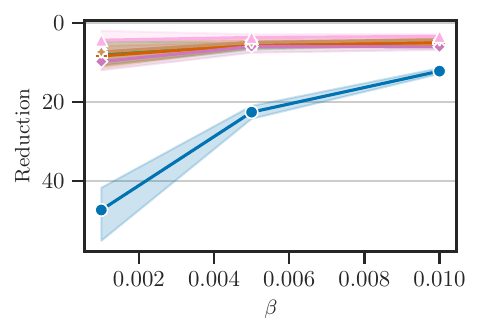}
        \caption{nDCG @ 20 ($\gamma=0.1$)}
        \label{fig:ndcg_0.1}
    \end{subfigure}
    \begin{subfigure}[t]{0.43\textwidth}
        \centering
        \includegraphics[width=\linewidth]{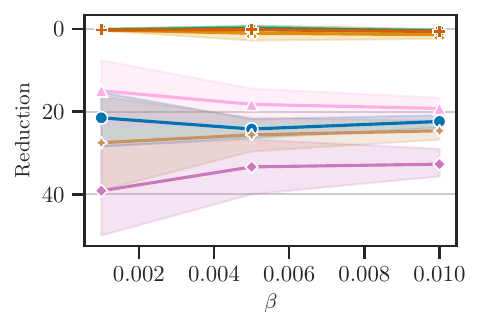}
        \caption{Recall @ 20 ($\gamma=0.001$)}
        \label{fig:recall_0.001}
     \end{subfigure}
    \begin{subfigure}[t]{0.43\textwidth}
        \centering
        \includegraphics[width=\linewidth]{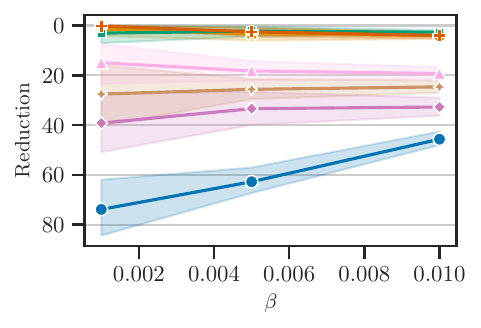}
        \caption{Recall @ 20 ($\gamma=0.01$)}
     \end{subfigure}
    \begin{subfigure}[t]{0.43\textwidth}
        \centering
        \includegraphics[width=\linewidth]{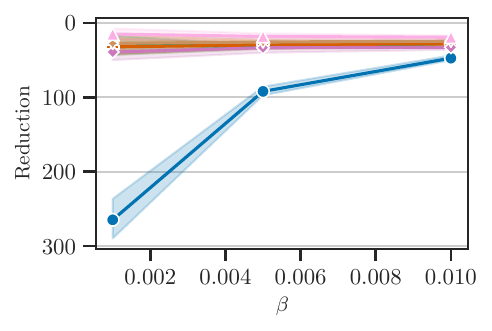}
        \caption{Recall @ 20 ($\gamma=0.1$)}
        \label{fig:recall_0.1}
     \end{subfigure}
     \caption{
        Expected performance reduction as a function of the fraction of adversarial users in the calibration set, $\beta \in [0.001, 0.1]$, for different reporting rates $\gamma \in \{0.001, 0.01, 0.1\}$ and reporting strategies.
        A reduction between 0 and 1 indicates an effect proportional to the size of the collective, whereas values greater than 1 indicate a disproportionate impact.
        Shaded areas denote confidence intervals obtained via bootstrapping over 10 runs.
     }
     \label{fig:performance_reduction}
\end{figure}

\paragraph{\textbf{(RQ1) Effective adversarial strategies should target low-risk items.}}
For a fixed collective size $\beta$, \cref{fig:reporting_strategies} shows that all considered strategies can maintain a high expected adversarial risk at calibration time $r_\lambda^{adv}$ up to relatively large filtering thresholds ($\lambda \approx 0.75$), across all reporting rates $\gamma \in \{0.001, 0.01, 0.1\}$.
Among them, the \texttt{LowRisk} strategy consistently achieves the highest empirical risk, even when adversarial users report as little as $0.1\%$ of the items they encounter.
This outcome is expected: by targeting low-risk items, adversaries ensure that these items are removed only at higher thresholds, forcing the system to adopt more conservative filtering.
The \texttt{Tags} strategy also performs competitively, where its effectiveness is comparable to \texttt{LowRisk} with a low reporting budget ($\gamma = 0.001$).
However, \cref{fig:fraction_reported_items_per_strategy} shows that \texttt{Tags} requires reporting roughly two orders of magnitude more items to achieve a similar empirical risk, highlighting the importance of selecting the \emph{right} items rather than merely increasing reporting volume.
By contrast, the \texttt{Likes} and \texttt{TopRanker} strategies, despite reporting the same number of items as \texttt{LowRisk}, result in approximately an expected risk that is an order of magnitude lower, performing similarly to random reporting.
We attribute this to two factors:
first, in the two-stage architecture, items ranked highly by the recommender are not necessarily those with the lowest estimated risk and may therefore be filtered early as $\gamma$ increases, limiting the effectiveness of \texttt{TopRanker}; second, item popularity is only an imperfect proxy for unwantedness; in \textit{Kuaishou}, likes might not reliably reflect perceived harm, which reduces the effectiveness of \texttt{Likes}.
Nonetheless, \cref{fig:reporting_strategies_gamma_0.001} shows that at $\gamma = 0.001$ (a reporting rate close to what is observed in practice) and $\gamma = 0.01$, the \texttt{Likes} strategy improves over \texttt{TopRanker} and \texttt{Random}.
We argue that the relatively modest performance of \texttt{Likes} in our setting is likely amplified by the fact that, in the implementation by \citet{detoni2025you}, risk scores are learned solely from sparse user-item interactions without incorporating auxiliary features.
In deployments where richer signals are available, popularity-based proxies such as likes may therefore allow adversarial collectives to more closely approximate the whitebox \texttt{LowRisk} strategy.
Finally, \cref{fig:reporting_strategies_gamma_0.1} shows that when the reporting rate increases to $\gamma = 0.1$, all strategies converge toward the performance of \texttt{Tags}.
This result highlights that, without precise information about item risk, reporting a sufficiently large number of items increases the likelihood that flagged items remain in the top-$k$ recommendations during calibration, thereby sustaining adversarial impact.

\begin{figure}
    \centering
    \begin{subfigure}[t]{0.7\textwidth}
        \centering
            \includegraphics[width=\linewidth]{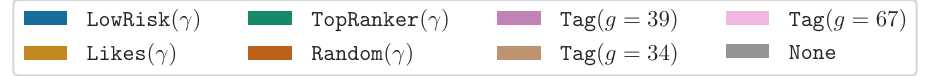}
            \vspace{0.1em}
    \end{subfigure}
    \begin{subfigure}[t]{0.49\textwidth}
         \centering
            \includegraphics[width=\linewidth]{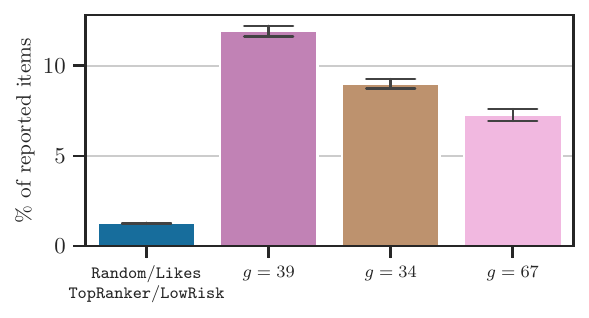}
            \caption{Frequency of reported items.}
            \label{fig:fraction_reported_items_per_strategy}
     \end{subfigure}
     \begin{subfigure}[t]{0.49\textwidth}
         \centering
            \includegraphics[width=\linewidth]{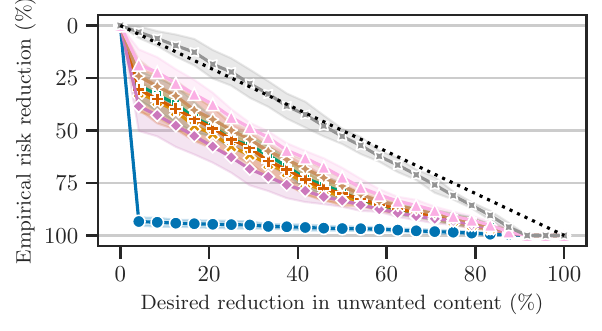}
          \caption{Expected risk for standard users.}
          \label{fig:expected_harmfulness_standard_users}
     \end{subfigure}
    \caption{
    Relationship between reporting intensity and risk reduction.
    A small fraction of carefully selected reported items in the calibration set (\cref{fig:fraction_reported_items_per_strategy}) can induce a sharp reduction in expected unwanted content for non-adversarial users relative to the baseline (\texttt{None} in \cref{fig:expected_harmfulness_standard_users}).
    Across all panels, the collective size is fixed to $\beta = 0.01$ and the reporting rate to $\gamma = 0.1$.
    We report the standard deviation over 10 runs as a shaded area or error bars.
    }
\end{figure}

\paragraph{\textbf{(RQ2) Few adversaries can cause a disproportionate impact on performance.}}
In this experiment, we fix a target reduction in expected risk of $25\%$ and evaluate the resulting \textit{performance reduction} (\cf \cref{eqn:reduction-ATE}) on nDCG@k and Recall@k for non-adversarial users at test time.
As shown in \cref{fig:performance_reduction}, even very small adversarial collectives can induce disproportionate degradations in recommendation quality when employing effective strategies.
In particular, \cref{fig:ndcg_0.001,fig:recall_0.001} show that, under a low reporting budget ($\gamma = 0.001$), the \texttt{LowRisk} strategy yields reductions of up to 5 and 20 in nDCG@20 and Recall@20, respectively, with a collective comprising only $0.1\%$ of calibration users (approx. four users).
When the reporting budget increases to $\gamma = 0.1$, the same strategy produces substantially larger effects, with reductions reaching up to 40 for nDCG@20 and 300 for Recall@20.
\gdt{As a concrete example, for $\beta=0.002$ and $\gamma=0.1$, we observe nDCG@20 reductions of approximately $10\%$ and $2\%$ under the \texttt{LowRisk} and \texttt{Likes} strategies, respectively. 
For Recall@20, the corresponding reductions are substantially larger, at $60\%$ and $10\%$.
While these changes may appear modest -- particularly for nDCG -- prior work shows that even small variations in recommender performance can translate into significant business impact~\cite{jannach2019measuring} (\eg a 2--3\% increase in click-through rate on eBay was associated with a 6\% increase in revenue~\cite{brovman2016optimizing}).}
\gdt{As detailed in \cref{sec:reporting_strategies}, the \texttt{LowRisk} strategy highlights a potential upper-bound in performance degradation across diverse collective sizes.}
By contrast, for $\gamma < 0.01$, the remaining strategies struggle to generate appreciable performance losses.
However, once adversarial users report $10\%$ of the items they encounter, all strategies lead to noticeable degradation, with reductions of up to 10 and 40 for nDCG@20 and Recall@20, respectively.

Interestingly, as shown in \cref{fig:expected_harmfulness_standard_users}, non-adversarial users simultaneously experience a sharp reduction in expected risk relative to the baseline (grey line in \cref{fig:expected_harmfulness_standard_users}), depending on the adversarial strategy adopted.
This highlights a key mechanism underlying the observed performance loss.
Risk-controlling recommender systems enforce risk constraints by filtering risky items and \textit{replacing them with safe alternatives} -- often previously consumed items that have not been flagged as unwanted \cite[Property~1]{detoni2025you}.
\gdt{This mechanism is consistent with real-world recommender ecosystems, where repeated consumption is common across domains and platforms routinely resurface previously consumed items or unfinished content~\cite{anderson2014dynamics,li2023repetition}.}
While effective for risk reduction, excessive reliance on repeated items undermines novelty and serendipity, thereby degrading recommendation quality.
Consistent with this mechanism, \cref{fig:fraction_repeated_items_per_strategy} shows that adversarial strategies can dramatically increase content repetition; under the \texttt{LowRisk} strategy, up to $80\%$ of recommended items have been seen previously by users, leading to sustained performance degradation.
Overall, these results suggest that adversarial effectiveness depends critically on the relative strength of adversarial signals compared to those of non-adversarial users.
Given that most users exhibit extremely low expected risk, even a small collective that injects a stronger, coordinated signal can exert an outsized influence, an observation consistent with prior work on algorithmic collective action \cite{baumann2024algorithmic,hardt2023algorithmic}.

\begin{figure}
    \centering
    \begin{subfigure}[t]{0.7\textwidth}
        \centering
            \includegraphics[width=\linewidth]{figures/02/classic_legend.pdf}
            \vspace{0.1em}
    \end{subfigure}
    \begin{subfigure}[t]{\textwidth}
        \centering
            \includegraphics[width=0.55\linewidth]{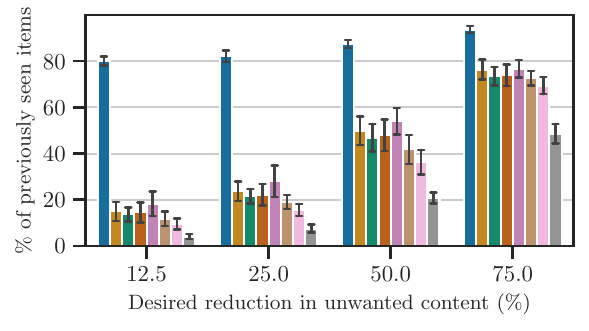}
    \end{subfigure}
    \caption{Adversarial strategies substantially increase content repetition in recommendations, leading to degraded performance -- up to $80\%$ repeated items under the \texttt{LowRisk} strategy.
    The collective size is fixed to $\beta = 0.01$ and the reporting rate to $\gamma = 0.1$.
    We report the standard deviation over 10 runs as error bars.}
    \label{fig:fraction_repeated_items_per_strategy}
\end{figure}

\paragraph{\textbf{(RQ3) Adversarial reporting cannot selectively alter group exposure.}}
Building on the previous findings, we investigate whether a coordinated adversarial collective can manipulate the exposure of specific item groups through targeted reporting, as implemented by the \texttt{Tag} strategy.
To this end, we focus on group $g = 34$ and compare its exposure under two adversarial collectives: one using the \texttt{Random}$(\gamma)$ strategy and another using the targeted \texttt{Tag}$(g = 34)$ strategy.
To ensure a fair comparison, we set $\gamma$ for \texttt{Random} to match the fraction of items reported under \texttt{Tag}$(g = 34)$ (approx. 10\% in \cref{fig:fraction_reported_items_per_strategy}).
Further, we evaluate the difference in exposure across varying collective sizes $\beta \in \{0.001, 0.005, 0.01\}$.
As shown in \cref{fig:difference_between_tags}, we observe no appreciable difference in exposure between the two strategies at any collective size.
This result follows directly from the design of the risk metric introduced by \citet{detoni2025you}, which is agnostic to item group membership and penalizes items solely based on their estimated risk score.
\gdt{The risk metric (\cref{eqn:set-based-harmfulness}) depends only on the number of flagged items within each top-$k$ list and is agnostic to which specific items are flagged. The same holds for the threshold calibration procedure.}
Consequently, adversaries cannot selectively reduce the exposure of \textit{any} target group simply by reporting items from that group more frequently.
\gdt{As further discussed in \cref{app:disparate-impact}, disparities in exposure may still arise, but only due to pre-existing biases in the underlying recommender model, and would diminish under an optimal risk predictor.
Importantly, even under such ideal conditions, \textit{coordinated users can still influence recommendation outcomes}: in risk-controlling recommender systems, the threshold $\lambda$ is calibrated from user feedback, independently of the scoring model, as conformal risk control guarantees are model-agnostic.}
Overall, these findings indicate that risk-controlling recommender systems exhibit a \textit{degree of robustness} to adversarial attempts at manipulating group-level exposure.

\begin{figure}[t]
    \centering
    \includegraphics[width=0.55\linewidth]{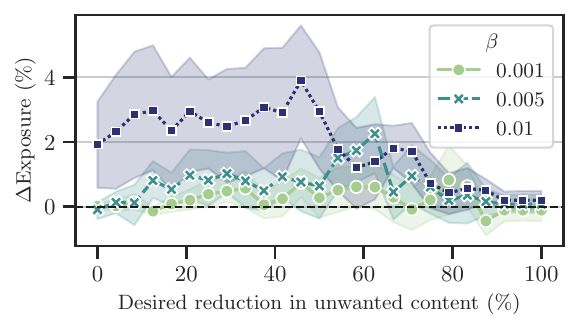}
     \caption{
     Difference in top-$k$ exposure ($k = 20$) for group $g = 34$ in \textit{Kuaishou}.
    Comparing \texttt{Random}$(\gamma)$ and \texttt{Tag}$(g = 34)$, both strategies flag approximately $10\%$ of items yet induce similar exposure reductions for non-adversarial users.
    Results are shown for collective sizes $\beta \in \{0.001, 0.005, 0.01\}$, with confidence intervals obtained via bootstrapping over 10 runs.
    }
    \label{fig:difference_between_tags}
\end{figure}

\section{Mitigating the Effects of Adversarial Collective Action}
\label{sec:mitigating-coll-action}
Our findings \gdt{in \cref{sec:evaluating_effects_of_collective}} suggest that \textit{risk-controlling recommender systems} are neither inherently fragile nor universally robust. Their vulnerability depends critically on the level at which guarantees are enforced and on how replacement mechanisms interact with user behaviour.
We argue that a more promising design direction is to move away from a \textit{single global threshold} toward \emph{user-level risk control}.
Concretely, instead of computing a single threshold $\lambda \in \Lambda$ (\cf \cref{eqn:recommendation-set-filter}) using a calibration set drawn from the entire population, the system could estimate personalised thresholds $\lambda_u$ for each user $u \in \mathcal{U}$.
\gdt{More formally, for a given user $u \in \calU$, let us assume we have access to a held-out calibration set $\calQ_u = \{(u,i,h)_j\}_{j=1}^{Q_{u}}$ of user-item interactions.
Then, the user-level threshold $\hat{\lambda}_u$ for a given risk $\alpha \in [0,1]$ can be defined as follow:
\begin{equation}
    \hat{\lambda}_u = \inf \left\{ \lambda : \frac{Q_u}{Q_u+1}\hat{R}_u(\lambda) + \frac{1}{Q_u+1} \leq \alpha \right\} \quad \text{s.t.} \quad \hat{R}_u(\lambda) = \frac{1}{Q_u}\sum_{j=1}^{Q_u} R(S_\lambda(U = u, k))
    \label{eqn:personalized-treshold}
\end{equation}
where $\hat{R}_u(\lambda)$ denotes the per-user risk score, computed only over the interactions made by the given user.}
By calibrating risk guarantees at the individual level, the effects of coordinated adversarial behaviour can be localized, preventing them from propagating system-wide.
Under such a design, the reporting strategies described in \cref{sec:reporting_strategies} \textit{would lose their effectiveness}, as reported items would influence only the reporting user’s personalised threshold.
\gdt{In this section, we empirically evaluate this approach by answering the following research question:} 
\begin{itemize}
    \item[\textbf{RQ4}] \gdt{Can we mitigate the adversarial collectives' influence on the recommender by employing user-level thresholds?}
\end{itemize}
\paragraph{Experimental details.} \gdt{In our experiment, we adopt the same empirical setup and evaluation protocol described in \cref{sec:evaluating_effects_of_collective}.
We focus on the \texttt{LowRisk} strategy, as it exhibits the (theoretical) largest impact on recommender performance. 
We compare two approaches: (i) a personalized calibration of the risk-controlling threshold (\cref{eqn:personalized-treshold}), and (ii) the standard population-level calibration used in \cref{sec:evaluating_effects_of_collective}.
To quantify the differences, we measure the performance gap $\Delta m$ for $m \in \{\text{nDCG@}k, \text{Recall@}k\}$ between the user-level and population-level risk-controlling recommender systems, where $\Delta > 0$ indicates an improvement of the user-level approach. 
Additionally, as in RQ3, we report the difference in the fraction of previously seen items appearing in the top-$k$ recommendations under the two calibration strategies. 
In this case, $\Delta < 0$ indicates that the user-level approach requires fewer previously seen items to achieve the same level of risk control.
All results are computed over non-adversarial users in the test set.}

\begin{figure}
    \centering
    \begin{subfigure}[t]{0.31\textwidth}
         \centering
            \includegraphics[width=\linewidth]{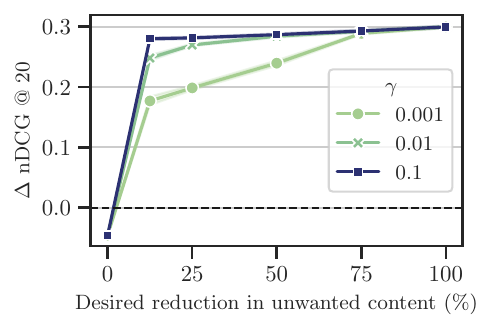}
          \caption{nDCG @ 20}
     \end{subfigure}
    \begin{subfigure}[t]{0.31\textwidth}
        \centering
        \includegraphics[width=\linewidth]{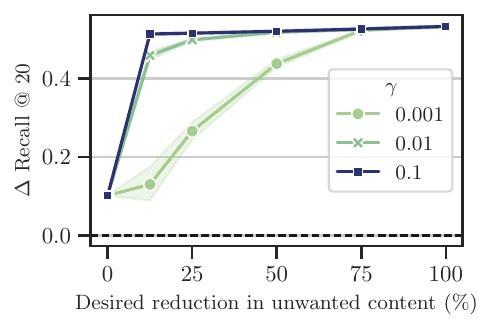}
        \caption{Recall @ 20}
     \end{subfigure}
     \begin{subfigure}[t]{0.32\textwidth}
         \centering
            \includegraphics[width=\linewidth]{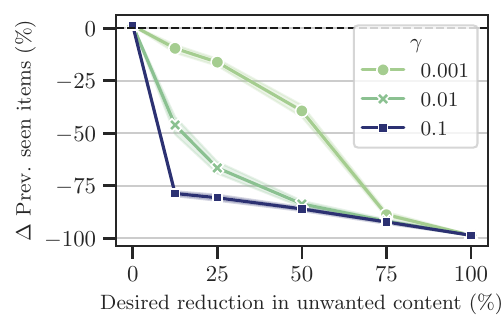}
          \caption{Previously seen items}
          \label{fig:previously-seen-items-difference}
     \end{subfigure}
     \caption{
        \gdt{Difference ($\Delta$) in performance between user-level and population-level thresholds in risk-controlling recommender systems at test time, under the \texttt{LowRisk} strategy.
        For nDCG @ 20 and Recall @ 20, a higher difference is better.
        For the previously seen items, a lower difference is better. 
        Results are shown for reporting rates $\gamma \in \{0.001, 0.01, 0.1\}$, with confidence intervals obtained via bootstrapping over 10 runs.
        The collective size is fixed to $\beta = 0.01$.}
     }
     \label{fig:effectiveness-strategy-per-user}
\end{figure}

\paragraph{\textbf{(RQ4) Personalized thresholds mitigate the effects of algorithmic collectives. }}
\gdt{Under the strongest attack strategy (\texttt{LowRisk}), adopting personalized thresholds as in \cref{eqn:personalized-treshold} yields consistent improvements over the standard global threshold. 
In particular, we observe gains of $+0.3$ nDCG @ 20 and $+0.4$ Recall @ 20, while reducing the number of previously seen (\ie repeated) items in the recommendation lists by approximately $99\%$ (\cref{fig:previously-seen-items-difference}).
\cref{fig:effectiveness-strategy-per-user} further illustrates these results. 
Importantly, these improvements are achieved without sacrificing the desired risk-control guarantees: both approaches satisfy the target bounds, although the global threshold tends to be more conservative.
The observed gains can be explained by the fact that personalized calibration prevents coordinated groups from influencing the calibration process.
Unlike the global threshold setting - where adversarial collectives can distort the shared calibration set (\cf \cref{sec:evaluating_effects_of_collective}) - each user now relies on an \textit{individual calibration set}, effectively isolating them from collective manipulation.
Moreover, personalization adapts the intervention to heterogeneous user behavior. 
Since most users exhibit low reporting rates, even under stringent reduction targets $\alpha \in [0,1]$, only a small number of items need to be filtered or replaced (on average $\approx 1.5$ items per top-$k$ list). 
This leads to higher-quality recommendations while still enforcing risk control for any $\alpha$. 
Consistently, we find that the benefits of personalization increase with both the desired reduction in unwanted content and the strength of the collective attack, indicating that user-level calibration is particularly effective in mitigating coordinated adversarial behavior.}

\section{Discussion}
\label{sec:discussion}

We now discuss some limitations of the present research, which open up further avenues for future work. 
First, our empirical evaluation necessarily relies on an \textit{offline simulation} with \textit{simulated collective of users} that may not reflect the complex interaction of a live system.
A natural next step is to validate these findings through online studies, similar to prior work using \textit{browser extensions} or \textit{participatory infrastructures}~\cite{mendler2024engine,hoang2021tournesol}, which would enable the collection of real user interactions with a functioning risk-controlling recommender system.
Second, the effectiveness of the proposed adversarial strategies depends on specific implementation choices and assumptions underlying the risk-controlling mechanism.
For instance, although \texttt{LowRisk} emerges as the most effective strategy in our experiments (\cf \cref{fig:performance_reduction}), it may be impractical to deploy in real-world settings due to limited access to internal risk scores.
That said, \textit{model extraction} and \textit{black-box inference techniques}~\cite{yue2021blackboxattacks} could, in principle, be used to approximate these scores, potentially making such strategies feasible.
More broadly, the impact of all strategies may differ in deployment, where feedback delays, interface design, and heterogeneous user behaviour could substantially alter outcomes.

In particular, users may employ the \textit{``Not Interested''} signal to express dissatisfaction with repetitive or low-novelty content rather than genuine disinterest or perceived harm~\cite{hong2025social}.
For example, prior work shows that tolerance to harmful content varies across cultural and social contexts \cite{jiang2021understanding}.
Given the current design of risk-controlling recommender systems and our findings on content repetition (\cf \cref{fig:fraction_repeated_items_per_strategy,fig:previously-seen-items-difference}), such well-intentioned feedback could inadvertently trigger feedback loops that push the system toward increasingly conservative filtering.
As a result, a small group of high-reporting users, acting in good faith, may disproportionately influence the system-wide risk threshold, effectively driving the recommender toward more conservative behaviour, akin to the adversarial effects analysed in \cref{sec:evaluating_effects_of_collective}.
\gdt{However, as we have shown empirically in \cref{sec:mitigating-coll-action}, employing user-level thresholds might represent a promising solution towards mitigating these phenomena.}
Beyond robustness, user-level risk control also offers a principled approach to addressing the \textit{cold-start} problem~\cite{park2009pairwise}.
New users could be initialised with predefined risk tolerances that get subsequently refined as feedback accumulates.
Such personalisation could further leverage existing social or relational structures on the platform, for example, by assigning new users the risk tolerance of socially connected or behaviourally-similar users, aligning risk control with established mechanisms of trust and similarity, and potentially fostering \textit{user altruism} \cite{fedorova2025user}.
Importantly, in this way, risk control ceases to be a purely \textit{platform-centric safeguard} to become a more \textit{user-centric approach} that empowers individuals to directly shape their recommendation experience.

\section{Conclusion}
This study investigates the vulnerability of \textit{risk-controlling recommender systems} to collectives of adversarial users. 
While conformal risk control provides principled, model-agnostic guarantees to bound unwanted content, we demonstrate that coordinated groups can exploit these mechanisms to their advantage.
Theoretically, we proved that adversarial users consume a portion of the global risk budget, forcing the system to adopt more conservative filtering thresholds for the entire population.
Empirically, using real-world data from a large video-sharing platform, we demonstrated that a collective of just $1\%$ of users can degrade recommendation quality for non-adversarial users by up to $20\%$.
We examined several realistic attack strategies that coordinated adversaries could adopt, including targeting top-ranked (\texttt{TopRanker}) and most-liked (\texttt{Likes}) items, and quantified their effects on widely used platform metrics such as nDCG@20 and Recall@20.
Targeting items with low estimated risk scores (\texttt{LowRisk}) yielded the strongest and most sustained adversarial impact; however, this strategy would require white-box access to the system, limiting its practical feasibility.
\gdt{Furthermore, our results reveal that these systems do not allow adversaries to selectively suppress specific item groups.}
\gdt{We also show that calibrating risk-controlling recommender systems with user-level thresholds can empirically mitigate the impact of adversarial collectives, leading to more stable performance as the collective size increases.}
Our findings demonstrate that global-level risk control is susceptible to coordinated manipulation, whereas user-level approaches offer a promising path forward to mitigate the impact of adversarial behaviour.

\section*{Acknowledgements}
The work of GDT and BL was partially supported by the following projects: Horizon Europe Programme, grants \#101120237-ELIAS and \#101120763-TANGO.
Funded by the European Union. Views and opinions expressed are however those of the author(s) only and do not necessarily reflect those of the European Union or the European Health and Digital Executive Agency (HaDEA). Neither the European Union nor the granting authority can be held responsible for them.

\bibliographystyle{ACM-Reference-Format}
\bibliography{references}

\appendix
\section{Proofs}
\label{app:proofs}

\subsection{Proof for \cref{theorem:adv-impact}}

\begin{proof}
Consider a held-out calibration set $\calQ = \{(u,i,h)_j\}_{j=1}^Q$.
Let us assume that there are $\calK \subset \calQ$ users that are behaving strategically.
Given $\lambda \in \Lambda$, let us denote the following expected risks for both standard and adversarial users on the calibration set:
\begin{align}
    \hat{R}(\lambda) &= \frac{1}{|\calQ|}\sum_{u' \in \calQ} R(S_\lambda(U=u', k)) \\
    &= \frac{1}{|\calQ|} \left[\sum_{u' \in \calQ\backslash\calK} R(S_\lambda(U=u', k)) + \sum_{u'' \in \calK} R(S_\lambda(U=u'', k)) \right] \\
    &= \frac{1}{|\calQ|}\sum_{u' \in \calQ\backslash\calK} R(S_\lambda(U=u', k)) +  \frac{1}{|\calQ|}\sum_{u'' \in \calK} R(S_\lambda(U=u'', k))
\end{align}
Let us denote $\hat{R}^{nonadv}(\lambda) = \frac{1}{|\calQ|}\sum_{u' \in \calQ\backslash\calK} R(S_\lambda(U=u', k))$.
Further, let us assume that all adversaries have the same risk $R(S_\lambda(U=u, k)) = r^{adv}$ for all $u \in \calK$. 
Thus, leading to the following equality:
\begin{equation}
    \hat{R}(\lambda) = \hat{R}^{nonadv}(\lambda) + \frac{|\calK|}{|\calQ|}r^{adv}
\end{equation}
Given a target level $\alpha \in [0,1]$, let us consider $\hat{\lambda} \in \Lambda$ choosen as $\hat{\lambda} = \inf \left\{ \lambda : \frac{Q}{Q+1}\hat{R}(\lambda) + \frac{1}{Q+1} \leq \alpha \right\}$, thus satisfying the risk control guarantees \cite{angelopoulos2024conformal}.
Let us denote with $Q = |\calQ|$ and $K = |\calK|$. 
Then, we can rewrite the selection threshold to distinguish between standard and adversarial users as follows:
\begin{align}
    &\frac{Q}{Q+1}\hat{R}(\lambda) + \frac{1}{Q+1} \leq \alpha \\
    &\frac{Q}{Q+1}\left(\hat{R}^{nonadv}(\lambda) + \frac{K}{Q}r^{adv} \right) + \frac{1}{Q+1} \leq \alpha \\
    & \frac{Q}{Q+1}\hat{R}^{nonadv}(\lambda) + \frac{\cancel{Q}}{Q+1} \frac{K}{\cancel{Q}}r^{adv} + \frac{1}{Q+1} \leq \alpha \\
    & \frac{Q}{Q+1}\hat{R}^{nonadv}(\lambda) + \frac{1}{Q+1} \leq \alpha - \frac{K}{Q+1}r^{adv} \label{eq:nonadv-bound}
\end{align}
In practice, the $K$ adversarial users ``consumes'' $\frac{K}{Q+1}$ units of risk budget, leaving only $\alpha - \frac{K}{Q+1}r^{adv}$ for the original population.
Thus, we can denote the quantity $\alpha_{\mathrm{eff}} = \alpha - \frac{K}{Q+1}r^{adv}$ as the \textit{effective risk reduction} for standard users. 

Since risks are nonnegative, we may take $\max\{0,\alpha_{\mathrm{eff}}\}$ as a valid target.
By exchangeability and the standard conformal risk control guarantee, inequality
\eqref{eq:nonadv-bound} implies that for a non-adversarial test user $U_{\mathrm{nonad}}$ we have
\begin{equation}
    \bbE[R(S_\lambda(U_{nonadv}, k))] \leq \max\{0, \alpha - \frac{K}{Q+1}r^{adv}_{\lambda}\}
\end{equation}
This proves Theorem~1.
\end{proof}

\section{Disparate Exposure in Risk-Controlling Recommenders}
\label{app:disparate-impact}
\gdt{In the experiments of \cref{sec:evaluating_effects_of_collective}, we examined the aggregate impact of adversarial collectives on recommendation performance.
We now turn to a finer-grained analysis of \emph{item-group exposure}.
Specifically, we focus on groups of items that share a similar ground-truth likelihood of being flagged as unwanted, \ie with comparable values of $\mathbb{E}[H = 1 \mid G]$.
For these groups, we measure empirical exposure at test time (\cf \cref{eqn:expected-group-exposure}) under the risk-controlling recommender system \emph{in the absence of adversarial users}.
While many groups experience comparable reductions in exposure--as expected, given that risk control replaces risky items with safer alternatives--we observe notable disparities across some groups.
For example, \cref{fig:difference_between_tags_no_adv} reports the exposure reduction for two item groups, $g \in \{54, 23\}$, which share the same ground-truth unwantedness rate ($\mathbb{E}[H = 1 \mid G = 54] \approx \mathbb{E}[H = 1 \mid G = 23] \approx 0.002$).
Despite this similarity, the groups exhibit markedly different exposure dynamics as the level of risk control varies.
In particular, at a target reduction of 50\% in unwanted content, the exposure of group 54 remains relatively stable, whereas group 23 experiences nearly a $30\%$ decrease.
This disparity arises from differences in the distributions of the learned risk scores across groups, as illustrated in \cref{fig:empirical-score-density}.
As a result, even when groups are equally likely to be flagged as unwanted, risk-controlling recommender systems may induce \emph{disparate exposure} \cite{pmlr-v81-ekstrand18b} if the underlying risk predictor $r(i,u)$ exhibits group-dependent biases.
Importantly, these disparities emerge independently of adversarial behaviour, highlighting that risk control might amplify pre-existing biases in the learned risk model rather than introducing new ones.
}

\begin{figure}[t]
          \centering
    \begin{subfigure}[t]{0.49\textwidth}
          \centering
        \includegraphics[width=\linewidth]{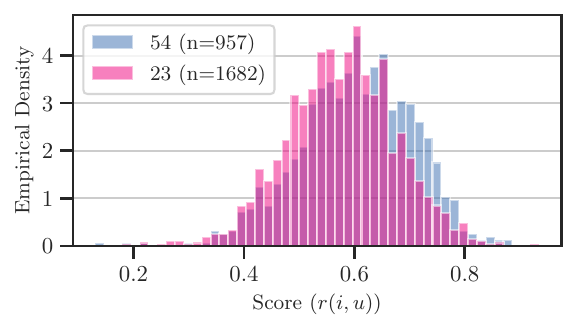}
        \caption{Empirical density of the risk scores}
        \label{fig:empirical-score-density}
     \end{subfigure}
     \begin{subfigure}[t]{0.49\textwidth}
         \centering
           \includegraphics[width=\linewidth]{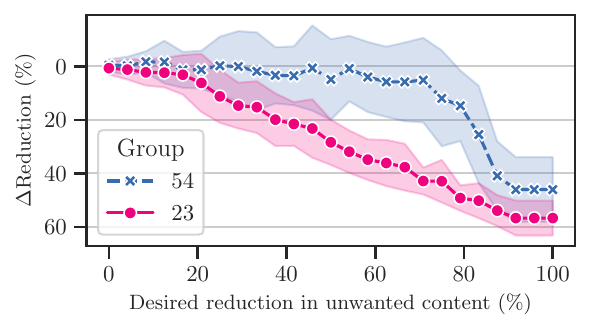}
          \caption{Reduction in exposure in top-k}
        \label{fig:difference_between_tags_no_adv}
     \end{subfigure}
   
    \caption{Exposure reduction for two item groups ($g \in \{54, 23\}$) in \textit{Kuaishou}.
    Although the two groups exhibit similar expected harmfulness, they differ in their risk score distributions: $\mathbb{E}[r(U,I) \mid G = 54] = 0.604$ and $\mathbb{E}[r(U,I) \mid G = 23] = 0.576$ (\cref{fig:empirical-score-density}).
    As a result, risk-control induces disparate exposure across the two groups (\cref{fig:difference_between_tags_no_adv}), even in the \textit{absence of adversarial users}.
    }
\end{figure}

\end{document}